\documentclass[envcountsame,11pt,letter]{llncs}

\usepackage{amssymb,amsmath,soul}
\usepackage{tikz}
\usetikzlibrary{arrows,automata,shapes}
\usepackage{enumerate}
\usepackage{xspace}

\newcommand{\aut}[1]{{\mathcal #1}}
\newcommand{\dual}[1]{{\mathfrak d}({#1})}
\newcommand{\inverse}[1]{{#1}^{-1}}
\newcommand{\mot}[1]{{\mathbf {#1}}}
\newcommand{\pres}[1]{\langle{#1}\rangle}
\newcommand{\presm}[1]{\pres{{#1}}_{+}}
\newcommand{\AAA}{A}
\newcommand{\MMM}{M}
\newcommand{\N}{{\mathbb N}}
\newcommand{\Z}{{\mathbb Z}}

\newcommand{\xxrightarrow}[1]{{\xrightarrow{~#1~}}}

\usepackage{color}

\newlength{\textlarg}

\newcounter{thesame}
\title{A Characterization of those Automata\\
that Structurally Generate Finite Groups}
\author{Ines Klimann$^\fnsymbol{thesame}$ \and Matthieu Picantin\thanks{Both authors are
partially supported by the french
\emph{Agence Nationale pour la~Recherche}, through the Project MealyM
ANR-JCJC-12-JS02-012-01.}}
\institute{Univ Paris Diderot, Sorbonne Paris Cit\'e, LIAFA,\\
    UMR 7089 CNRS, F-75013 Paris, France\\
\email{\{klimann,picantin\}@liafa.univ-paris-diderot.fr}}

\begin{document}

\maketitle

\begin{abstract}
Antonenko and~Russyev independently have shown that any Mealy automaton with no cycles with exit---that is,
where every cycle in the underlying directed graph is a sink component---generates
a finite (semi)group, regardless of the choice of the production functions.
Antonenko has proved that this constitutes a characterization in the
non-invertible case and asked for the invertible case, which is proved
in this paper.
\end{abstract}


\section{Introduction}

The class of automata (semi)groups contains multiple interesting and complicated (semi)groups
with sometimes unusual features~\cite{bs}.

In the last decades, the classical decision problems have been investigated for such (semi)groups. The 
word problem is solvable using standard minimization
techniques, while the conjugacy problem is undecidable for automata groups~\cite{conjugacy}.
Of special interest for our concern here, the finiteness problem was proved to be undecidable
for automata semigroups~\cite{gil13} and remains open for automata groups
despite several positive and promising results~\cite{AKLMP12,anto,sidkiconjugacy,cain,Kli13,KMP12,mal,min,sidki,sst}.

The family of automata with no cycles with exit was investigated by~Antonenko and by~Russyev
independently. Focused on the invertible case, Russyev stated in~\cite{russ}
that any invertible Mealy automata with no cycles with exit generates a finite group.
Meanwhile, Antonenko showed in~\cite{anto} (see also~\cite{antoberk})
the same result in the non-invertible case and proved the following maximality result:
for any automaton with at least one cycle with exit, it is possible to choose
(highly non-invertible) production functions such that the semigroup generated
by the induced Mealy automaton is infinite.

In this paper, we fill the visible gap by extending
the aforesaid maximality result to the invertible case:
for any automaton with at least one cycle with exit, it is possible to choose
invertible production functions such that the group generated
by the induced Mealy automaton is infinite.

The proof of this new result makes use of original arguments for the current framework,
whose common idea is to put a special emphasis on the \emph{dual automaton},
obtained by exchanging the roles of the stateset and the alphabet.
Thereby it continues to validate the general strategy first suggested
in the paper~\cite{AKLMP12}, then followed and continuously developed in~\cite{Kli13,KMP12}.

The new maximality result
provides a precious milestone in
the ongoing work by De~Felice and~Nicaud (see~\cite{DFN13} for a first paper)
who propose to design random generators for finite groups
based on those invertible Mealy automata with no cycles with exit.
Their aim is to simulate interesting distributions that might offer a wide diversity of different finite groups
by trying to avoid the classical concentration phenomenon around a typical object,
namely symmetric or alternating groups~\cite{Dix69,JZP11}, which is significant in already studied distributions.
Once implemented, such generators would be very useful to test the performance
and robustness of algorithms from computational group theory.
They would also be of great use when trying to check a conjecture, by testing it on various random inputs, since exhaustive tests are impossible due to a combinatorial explosion.

The structure of the paper is the following. Basic notions
  on Mealy automata and automaton (semi)groups are presented in
  Section~\ref{sec:mealy}. In Section~\ref{sec:maximality}, we
introduce new tools and prove the main result.

%%%%%%%%%%%%%%%%%%%%%%%%%%%%%%%%%%%%%%%%%%
\section{Mealy Automata}\label{sec:mealy}
This section contains material for the proofs: first classical definitions
and then considerations already made in~\cite{Kli13} to maintain
the paper self-contained.

\subsection{Automaton Groups and Semigroups}
If one forgets initial and final states, a {\em
(finite, deterministic, and complete) automaton} $\aut{\AAA}$ is a
triple 
\[
\bigl( A,\Sigma,\delta = (\delta_i: A\rightarrow A )_{i\in \Sigma} \bigr)\enspace,
\]
where the \emph{stateset}~$A$
and the \emph{alphabet}~$\Sigma$ are non-empty finite sets, and
where the $\delta_i$ 
are functions called \emph{transition functions}.

The transitions of such an automaton are
\[x\xxrightarrow{~~i~~} \delta_i(x)\enspace.
\]

An automaton is \emph{reversible} if all its transition functions are
permutations of the stateset. Note that in this case each state has
exactly one incoming transition labelled by each letter.

\smallskip

A \emph{Mealy automaton} is a quadruple 
\[
\bigl( A, \Sigma, \delta = (\delta_i: A\rightarrow A )_{i\in \Sigma},
\rho = (\rho_x: \Sigma\rightarrow \Sigma  )_{x\in A} \bigr)\enspace,
\]
such that both $(A,\Sigma,\delta)$ and $(\Sigma,A,\rho)$ are
automata. 
In other terms, a Mealy automaton is a letter-to-letter transducer
with the same input and output alphabet. If
\(\aut{\AAA}=(A,\Sigma,\delta)\) is an automaton and \(\rho=(\rho_x:
\Sigma\rightarrow \Sigma)_{x\in A}\) is a finite sequence of
functions, we denote by~\((\aut{\AAA},\rho)\) the Mealy automaton
\((A,\Sigma,\delta,\rho)\) and we say that \(\aut{\AAA}\) is
\emph{enriched with}~\(\rho\).
The graphical representation of a Mealy automaton is
standard, see~Fig.~\ref{fig:Mealy-automaton}.

\begin{figure}[h]
\centering
\begin{tikzpicture}[->,>=latex,node distance=1.8cm]
\tikzstyle{every state}=[minimum size=12pt,inner sep=2pt]
\node[state] (0) {\(y\)};
\node[state] (1) [right of=0] {\(x\)};
\path (0) edge[loop left] node[above]{\(0|0\)} node[below]{\(1|1\)} (0)
      (1) edge node[above] {\(0|1\)} (0)
      (1) edge[loop right] node{\(1|0\)} (1);    
\end{tikzpicture}
\caption{An example of a Mealy automaton: the so-called adding machine.}\label{fig:Mealy-automaton}
\end{figure}
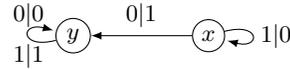

The transitions of a Mealy automaton are
\[x\xxrightarrow{i | \rho_x(i)} \delta_i(x)\enspace.
\]

A Mealy automaton \(\aut{\MMM}=(A,\Sigma,\delta, \rho)\) is
\emph{reversible\/} if the automaton \((A,\Sigma,\delta)\) is
reversible and \emph{invertible\/} if the functions
\(\rho_x\) are permutations of the alphabet. In this latter case, its
\emph{inverse} is the Mealy automaton~\(\inverse{\aut{\MMM}}\)
with stateset \(A^{-1}=\{x^{-1},x\in A\}\) and set of transitions
\begin{equation*}
x^{-1} \xxrightarrow{j\mid i} y^{-1} \ \in \aut{\MMM}^{-1} \quad \iff
\quad x \xxrightarrow{i\mid j} y \ \in \aut{\MMM} \:.
\end{equation*}
A Mealy automaton \(\aut{\MMM}\) is \emph{bireversible} if both
\(\aut{\MMM}\) and \(\inverse{\aut{\MMM}}\) are invertible and reversible.

\smallskip

In a Mealy automaton $\aut{\MMM}=(A,\Sigma, \delta, \rho)$, the sets $A$ and
$\Sigma$ play dual roles. So we may consider the \emph{dual (Mealy)
  automaton} defined by
\[
\dual{\aut{\MMM}} = (\Sigma,A, \rho, \delta)\enspace,
\] see an example on Fig.~\ref{fig:dual}.
Obviously, a Mealy automaton is reversible if and only if its dual is
invertible.

\begin{figure}[h]
\centering
\begin{tikzpicture}[->,>=latex,node distance=1.8cm]
\tikzstyle{every state}=[minimum size=12pt,inner sep=2pt]
\node[state] (0) {\(0\)};
\node[state] (1) [right of=0] {\(1\)};
\path (0) edge[loop left] node{\(y|y\)} (0)
      (1) edge[bend left] node[below] {\(x|x\)} (0)
      (0) edge[bend left] node[above] {\(x|y\)} (1)
      (1) edge[loop right] node{\(y|y\)} (1);    
\end{tikzpicture}
\caption{The dual automaton of the Mealy automaton of Fig.~\ref{fig:Mealy-automaton}.}\label{fig:dual}
\end{figure}
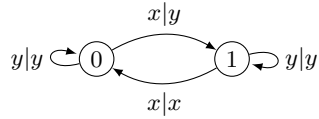

\smallskip

Let $\aut{\MMM} = (A,\Sigma, \delta,\rho)$ be a Mealy automaton. 
We view $\aut{\MMM}$ as an automaton with an input and an output tape, thus
defining mappings from input words over~$\Sigma$ to output words
over~$\Sigma$. 
Formally, for $x\in A$, the map
$\rho_x : \Sigma^* \rightarrow \Sigma^*$,
extending~$\rho_x : \Sigma \rightarrow \Sigma$, is defined by:
\begin{equation*}
\forall i \in \Sigma, \ \forall \mot{s} \in \Sigma^*, \qquad
\rho_x(i\mot{s}) = \rho_x(i)\rho_{\delta_i(x)}(\mot{s}) \enspace.
\end{equation*}

By convention, the image of the empty word is itself.
The mapping~$\rho_x$ is length-preserving and prefix-preserving.
We say that $\rho_x$ is the \emph{production
function\/} associated with \(x\) in~$\aut{\MMM}$ or, 
more briefly, if there is no ambiguity,
the \emph{production function\/} of \(x\).
For~$\mot{u}=x_1\cdots x_n \in A^n$ with~$n>0$, we set
\(\rho_\mot{u}: \Sigma^* \rightarrow \Sigma^*, \rho_\mot{u} = \rho_{x_n}
\circ \cdots \circ \rho_{x_1}\).

Denote dually by $\delta_i:A^*\rightarrow A^*,
i\in \Sigma$, the production functions associated with
the dual automaton
\(\dual{\aut{\MMM}}\). For~$\mot{s}=i_1\cdots i_n
\in \Sigma^n$ with~$n>0$, we set $\delta_\mot{s}: A^* \rightarrow A^*,
\ \delta_\mot{s} = \delta_{i_n}\circ \cdots \circ \delta_{i_1}$. 

\smallskip

The semigroup of mappings from~$\Sigma^*$ to~$\Sigma^*$ generated by
$\rho_x, x\in A$, is called the \emph{semigroup generated
by~$\aut{\MMM}$} and is denoted by~$\presm{\aut{\MMM}}$.
When $\aut{\MMM}$ is invertible,
its production functions are
permutations on words of the same length and thus we may consider
the group of mappings from~$\Sigma^*$ to~$\Sigma^*$ generated by
$\rho_x, x\in A$; it is called the \emph{group generated
by~$\aut{\MMM}$} and is denoted by~$\pres{\aut{\MMM}}$.

The automaton of Fig.~\ref{fig:Mealy-automaton} generates the
semigroup~\(\N\) and the group~\(\Z\). The orbit of the
word~\(0^n\) under the action of \(\rho_x\) is of size \(2^n\):
it acts like a binary addition until \(1^n\) (considering the most
significant bit on the right). In fact the Mealy automaton of
Fig.~\ref{fig:Mealy-automaton} is called the \emph{adding
  machine}~\cite{gns}.

\medskip

\newcommand{\semigroup}{(\({\mathbf{F2}}\))\xspace}
\newcommand{\prune}{(\({\mathbf{F1}}\))\xspace}
\newcommand{\duality}{(\({\mathbf{F3}}\))\xspace}
\newcommand{\justir}{(\({\mathbf{F4}}\))\xspace}

Remind some known facts on finiteness of the automaton (semi)group:
\begin{enumerate}[\((i)\)]
\item[\prune] 
To prune a Mealy automaton by deleting its states which are not
reachable from a cycle (see the precise definition of a cycle in
Subsection 3.1) does not change the finiteness or infiniteness of the
generated (semi)group.
\item[\semigroup] An invertible Mealy automaton generates a finite group if and
  only if it generates a finite semigroup~\cite{AKLMP12,svv}.
\item[\duality] A Mealy automaton generates a finite semigroup if and only if so
  does its dual~\cite{AKLMP12,nek,sv11}.
\item[\justir] An invertible-reversible but not bireversible Mealy automaton
  generates an infinite group~\cite{AKLMP12}.
\end{enumerate}

Whenever the alphabet is unary, the generated group is trivial and there is nothing to say. 
{\bf Throughout this paper, the alphabet has at least two elements.}

\subsection{On the Powers of a Mealy Automaton and its Connected Components}
Let \(\aut{\MMM}=(A,\Sigma,\delta,\rho)\) be a Mealy automaton.

Considering the underlying graph of \(\aut{\MMM}\), it makes sense
to look at its connected components. If \(\aut{\MMM}\) is reversible, its
connected components are always strongly connected (its transition
functions are permutations of a finite set).

A convenient and natural operation is to raise~\(\aut{\MMM}\)
to the power~\(n\), for some~\(n>0\): its \emph{\(n\)-th power} is the
Mealy automaton
\begin{equation*}
\aut{\MMM}^n = \bigl( \ A^n,\Sigma, (\delta_i : A^n \rightarrow
A^n)_{i\in \Sigma}, (\rho_{\mot{u}} : \Sigma \rightarrow \Sigma
)_{\mot{u}\in A^n} \ \bigr)\enspace.
\end{equation*}
If \(\aut{\MMM}\) is reversible, so is each of its powers.

\medskip

If \(\aut{\MMM}\) is reversible, we can be more precise on the behavior
of the connected components of its powers. As highlighted
in~\cite{Kli13}, they have a very peculiar form: if \(\aut{C}\) is a
connected component of~\(\aut{\MMM}^n\) for some~\(n\) and \(\mot{u}\) is a
state of~\(\aut{C}\), we obtain a connected component
of~\(\aut{\MMM}^{n+1}\) by choosing a state~\(x\in A\) and building the
connected component of~\(\mot{u}x\), denote it by~\(\aut{D}\). For any
state~\(\mot{v}\) of~\(\aut{C}\), there exists a state of~\(\aut{D}\)
prefixed with~\(\mot{v}\):
\[\exists\mot{s}\in\Sigma^*\mid \delta_{\mot{s}}(\mot{u}) =
\mot{v}\quad \text{and}\quad  \delta_{\mot{s}}(\mot{u}x) =
\mot{v}\delta_{\rho_{\mot{u}}(\mot{s})}(x)\enspace.\]

Furthermore, if \(\mot{u}y\) is a state of~\(\aut{D}\), for some
state~\(y\in A\) different from~\(x\), then \(\delta_{\mot{s}}(\mot{u}x)\) and
\(\delta_{\mot{s}}(\mot{u}y)\) are two different states of~\(\aut{D}\)
prefixed with~\(\mot{v}\),
because of the reversibility of~\(\aut{\MMM}^{n+1}\): the transition
function~\(\delta_{\rho_{\mot{u}}(\mot{s})}\) is a permutation.

Hence \(\aut{D}\) can be seen as consisting of several full
copies of~\(\aut{C}\) and \(\#\aut{C}\) divides~\(\#\aut{D}\). They
have the same size if and only if, once fixed some state~\(\mot{u}\)
of~\(\aut{C}\), for any different states~\(x,y\in A\), \(\mot{u}x\) and
\(\mot{u}y\) cannot both belong to \(\aut{D}\).

If all of those connected components of~\(\aut{\MMM}^{n+1}\) built
from~\(\aut{C}\) have the same size as~\(\aut{C}\), we say that
\(\aut{C}\) \emph{splits up totally}. If all the connected components
of an automaton split up totally, we say that the automaton
\emph{splits up totally}.

\section{A Maximal Family for Groups}\label{sec:maximality}

Antonenko and Russyev both investigated a family
of Mealy automata such that the finiteness of the generated
group (Russyev~\cite{russ}) or semigroup (Antonenko~\cite{anto}) is
inherent to the structure of the automaton, 
regardless of its production functions. In fact, though they use
different definitions and names, they study the same
family: automata where every cycle is a sink component. Antonenko has
proved that this family is maximal in the non-invertible case: if an automaton
admits a cycle which is not a sink component, it can be
  enriched to generate an infinite semigroup.

To prove his result, Antonenko analyzes different cases and, in each
situation, exhibits an element of infinite order in the semigroup. 
In this section we prove the maximality of the former family for
groups, using completely different techniques.
We adopt and adapt Russyev's nomenclature.

\subsection{How to Exit from a Cycle?}\label{subs:cycle}
Let \(\aut{A} = (A,\Sigma,\delta)\) be an automaton. A \emph{cycle} of
length \(n\in\N\) in the
automaton~\(\aut{A}\) is a sequence of transitions of~\(\aut{A}\)
\[x_1\xxrightarrow{i_1}x_2,\quad\ldots,\quad
x_{n-1}\xxrightarrow{i_{n-1}}x_n, \quad
x_n\xxrightarrow{i_n}x_1\]
where \(x_1,\ldots,x_n\) are 
pairwise different states in~\(A\) and \(i_1,\ldots,i_n\) are
some letters of~\(\Sigma\).

The \emph{label} of this cycle \emph{from the state \(x_k\)} is the word
\(i_k\cdots i_ni_1\cdots i_{k-1}\).

This cycle is \emph{with external exit} if there
exist~\(k\) with~\(1\leq k\leq n\) and~\(i\in\Sigma\)
satisfying~\(\delta_i(x_k)\not \in\{x_1,\ldots,x_n\}\).
It is \emph{with internal exit} if there exist~\(k\) with~\(1\leq k\leq n\)
and~\(i\in\Sigma\) satisfying~\(\delta_i(x_k)\in\{x_1,\ldots,x_n\}\)
and~\(\delta_i(x_k)\neq\delta_{i_k}(x_k)\).
We could say that a cycle is \emph{with exit}
without specifying the nature of the exit. In all other cases, this cycle
is \emph{without exit}. Examples are given in Fig.~\ref{fig:cycles}.

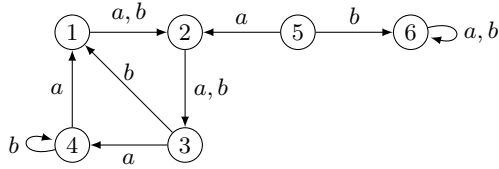
\begin{figure}[h]
\centering
\begin{tikzpicture}[->,>=latex,node distance=1.5cm]
\tikzstyle{every state}=[minimum size=12pt,inner sep=2pt]
\node[state] (1) {\(1\)};
\node[state] (2) [right of=1] {\(2\)};
\node[state] (3) [below of=2] {\(3\)};
\node[state] (4) [below of=1] {\(4\)};
\node[state] (5) [right of=2] {\(5\)};
\node[state] (6) [right of=5] {\(6\)};
\path (1) edge node[above] {\(a,b\)} (2)
      (2) edge node[right] {\(a,b\)} (3)
      (3) edge node[below] {\(a\)} (4)
      (3) edge node[above] {\(b\)} (1)
      (4) edge node[left] {\(a\)} (1)
      (4) edge[loop left] node {\(b\)} (4)
      (5) edge node[above] {\(a\)} (2)
      (5) edge node[above] {\(b\)} (6)
      (6) edge[loop right] node {\(a,b\)} (6);
\end{tikzpicture}
\caption{The cycle
  \(1\xrightarrow{a}2\xrightarrow{a}3\xrightarrow{b}1\) is with
  external exit; the cycle
  \(1\xrightarrow{a}2\xrightarrow{a}3\xrightarrow{a}4\xrightarrow{a}1\)
  is with internal exit; and the cycle \(6\xrightarrow{a}6\) is
  without exit.}\label{fig:cycles}
\end{figure}

Note that the existence of a cycle with internal exit induces the
existence of a (possibly shorter) cycle with external exit. For example
in Fig.~\ref{fig:cycles}, the cycle
\(1\xrightarrow{a}2\xrightarrow{a}3\xrightarrow{a}4\xrightarrow{a}1\)
has two internal exits: \(4\xrightarrow{b} 4\) and~\(3\xrightarrow{b}
1\); the first one leads to the cycle \(4\xrightarrow{b} 4\) with
external exit, while the second one leads to the cycle
\(1\xrightarrow{a}2\xrightarrow{a}3\xrightarrow{b}1\) with external exit.

\medskip

\begin{proposition}[\cite{anto,russ}]
Whenever an automaton~\(\aut{A}\) admits no cycle with exit,
whatever choice is maid for the production functions \(\rho\), the
enriched automaton~\((\aut{A},\rho)\) generates
a finite (semi)group.
\end{proposition}

\medskip

\subsection{A Pumping Lemma for the Reversible Two-State Automata}
It is proved in~\cite[Lemma~10]{Kli13} that in the case of a
reversible Mealy automaton~\(\aut{\MMM}\) with exactly two states, if
some power of~\(\aut{\MMM}\) splits up totally, then all the later
powers of~\(\aut{\MMM}\) split up totally. We can deduce the following 
result which can be seen as a pumping lemma: if the generated
semigroup is infinite, sufficiently long paths can be considered in
the dual automaton to turn indefinitely in a cycle.

\begin{lemma}[Pumping Lemma]\label{prop:suite}
Let \(\aut{\MMM}\) be a reversible Mealy
automaton with two states~\(\{x,y\}\). The automaton~\(\aut{\MMM}\)
generates an infinite semigroup if and only if, for any integer~\(N\in\N\),
there exists a word~\(\mot{u}\in \{x,y\}^*\) of length at least~\(N\) such
that the states~\(\mot{u}x\) and~\(\mot{u}y\) belong to the same
connected component of~\(\aut{\MMM}^{|\mot{u}|+1}\).
\end{lemma}

\subsection{The Family of Automata with no Cycles with Exit is Maximal
  for Groups}
We prove here that any automaton which admits a cycle with exit
can be enriched in order to generate an infinite group.
We analyze several simple cases in Lemmas~\ref{lm:binary},
\ref{lm:no-return}, and~\ref{lm:ir} which contribute to prove the
general case of Theorem~\ref{th:extension}.

\begin{lemma}\label{lm:binary}
Any automaton over a binary alphabet with a cycle with
external exit can be enriched to generate an infinite group.
\end{lemma}

\begin{proof}
Let \(\aut{A}\) be an automaton over a binary alphabet~\(\{0,1\}\)
with a cycle~\(\aut{C}\) with external exit as shown in
Fig.~\ref{fig:new-cycle}.

\begin{figure}[h]
\centering
\begin{tikzpicture}[->,>=latex,node distance=1.5cm]
\tikzstyle{every state}=[minimum size=8pt,inner sep=0pt]
\node[state] (x) at (180:1.2) {\(x\)};
\node[state] (x2) at (240:1.2) {};
\node[state] (y) [left of=x] {\(y\)};
\node (C) at (45:1.45) {\(\aut{C}\)};
\path (x) edge node[above] {\(i\)} (y);
\draw[<-,dashed] (170:1.2) arc (170:-110:1.2);
\draw (x.south) arc (190:235:1.2);
\node (i1) at (205:1.4) {\(j\)};
\end{tikzpicture}
\caption{The cycle \(\aut{C}\) is with external exit:
  \(y=\delta_i(x)\), \(y\not\in\aut{C}\).}\label{fig:new-cycle}
\end{figure}
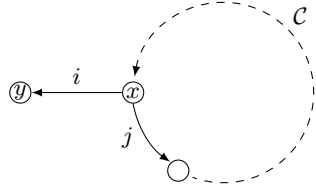

Take the following permutations on the alphabet : \(\rho_y\)
permutes the letters of the alphabet and \(\rho_z\) stabilizes the
alphabet for any other state~\(z\) (in particular for any state
of~\(\aut{C}\)).

Let \(\mot{s}\in\{0,1\}^+\) be the label of~\(\aut{C}\) from~\(x\).
For any~\(n\in\N\), the words~\(\mot{s}^ni0\) and~\(\mot{s}^ni1\)
belong to a same connected
component of~\(\dual{\aut{A},\rho}^{|\mot{s}|+2}\):
\(\rho_x(\mot{s}^ni0) = \mot{s}^ni1\). The Mealy automaton
\(\dual{\aut{A},\rho}\) is reversible and has two states, so we can
apply the Pumping Lemma and conclude on the infiniteness
of~\(\pres{(\aut{A},\rho)}\) by~\semigroup and~\duality.
\qed
\end{proof}

\begin{lemma}[River of no return Lemma]\label{lm:no-return}
Let~\(\aut{A}\) be an automaton and~\(\aut{C}\) a cycle
of~\(\aut{A}\). If \(\aut{C}\) admits an external exit to some state
and is not reachable from this state, then \(\aut{A}\) can be enriched
to generate an infinite group.
\end{lemma}

\begin{proof}
The idea of this proof is to mimic the adding
machine (see~Fig.~\ref{fig:Mealy-automaton}). Again, Fig.~\ref{fig:new-cycle} illustrates the situation:
\(\aut{C}\) admits an external exit to the state \(y\), the additional
hypothesis being that \(\aut{C}\) is not reachable from~\(y\) (and the alphabet is not supposed binary any longer).

Denote the label of~\(\aut{C}\) from~\(x\) by~\(\mot{s}=j\mot{t}\) with~\(j\in \Sigma\) and~\(\mot{t}\in \Sigma^*\).
We choose the following production functions on the alphabet: \(\rho_x\)
is the transposition of~\(i\) and~\(j\) and \(\rho_z\) is the identity
for any other state~\(z\).

As for the adding machine, the orbit of~\((j\mot{t})^n\) under the
action of~\(\rho_x\) has size~\(2^n\). Therefore the element~\(x\) is of infinite
order and so is the group~\(\pres{(\aut{A},\rho)}\).
\qed
\end{proof}

\begin{lemma}\label{lm:ir}
Any reversible automaton with a cycle with exit can be enriched to
generate an infinite group.
\end{lemma}

\begin{proof}
Let \(\aut{A}=(A,\Sigma,\delta)\) be a reversible automaton with a
cycle with exit.

As \(\aut{A}\) is reversible, it admits some states~\(x,
y, z\) with~\(x\neq y\) such that there exist a transition from~\(x\) to~\(z\)
and a transition from~\(y\) to~\(z\). We can choose the
permutations~\(\rho_x\) and~\(\rho_y\) such that these
transitions have the same output and take identity for all the other
permutations.

The enriched automaton \((\aut{A}, \rho)\) is invertible and
reversible but not bireversible. Hence it generates an infinite group by~\justir.
\qed
\end{proof}

The next theorem is the main result of this paper.

\begin{theorem}\label{th:extension}
Any automaton with a cycle with exit can be enriched into an invertible Mealy automaton generating an infinite group.
\end{theorem}

\begin{proof}
Let \(\aut{A}=(A,\Sigma,\delta)\) be an automaton with a cycle with exit.

By~\prune, we can suppose, without loss of generality,
that~\(\aut{A}\) is pruned.
If there exists a transition not belonging to a cycle, as the
starting state of this transition is reachable from a cycle,
Lemma~\ref{lm:no-return} applies and we are done.

We can assume now that any transition belongs to (at least) one cycle.

If \(\aut{A}\) is reversible, it can be enriched to generate an
infinite group by Lemma~\ref{lm:ir}.

\newcommand{\zzz}{y}
\newcommand{\yyy}{x'}
\newcommand{\ttt}{z}

\begin{figure}[h]
\centering
\begin{tikzpicture}[->,>=latex,node distance=1.8cm]
\tikzstyle{every state}=[minimum size=12pt,inner sep=2pt]
\node[state] (x) {\(x\)};
\node[state] (1) [right of=x] {};
\node[state] (2) [right of=1] {$\yyy$};
\node[state] (3) [right of=2] {$\zzz$};
\node (4) [right of=3] {};
\node[state] (5) [right of=4] {};
\path (x) edge node[above]{\(i\)} (1)
      (1) edge[dashed] (2)
      (2) edge node[above] {\(i\)} (3)
      (3) edge[dashed,-] (5)
      (4) edge node[above] {\(i\)} (5)
      (5) edge[bend left] node[above] {\(i\)} (3);
\end{tikzpicture}
\caption{Path from \(x\) with all transitions labelled
  by~\(i\).}\label{fig:icycle}
\end{figure}
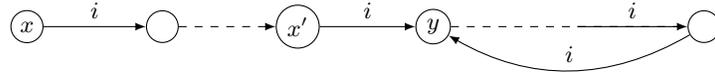

We can suppose now that \(\aut{A}\) is not reversible: there exist a state~\(x\)
and a letter~\(i\) such that \(x\) has no incoming transition labelled
by~\(i\).
Consider the path starting at~\(x\) with all
transitions labelled by~\(i\) as shown in Fig.~\ref{fig:icycle}.
This path loops on some state~\(\zzz\not=x\).
Denote by~\(\aut{C}\) the resulting cycle.
We have~\(\zzz=\delta_{i^{n+k\#\aut{C}}}(x)\)
for some minimal~\(n>0\) and for all~\(k\geq 0\).
Now let~\(\yyy\) denote the state~\(\delta_{i^{n-1}}(x)\): we have~\(\yyy\not\in\aut{C}\).

The transition~\(\yyy\xxrightarrow{i}\zzz\) belongs to some cycle 
by hypothesis and this cycle is not~\(\aut{C}\) by construction.
Therefore~\(\aut{C}\) admits an external exit~\(\zzz'\xxrightarrow{j}\zzz''\)
(with~\(j\not=i\)). Hence \(\yyy\) is reachable from~\(\zzz\) and so from~\(\aut{C}\),
by hypothesis, but does not belong to~\(\aut{C}\), by
construction. The automaton~\(\aut{B}=(A, \{i,j\}, (\delta_i,
\delta_j))\) contains the cycle~\(\aut{C}\) and the transition
\(\zzz'\xxrightarrow{j}\zzz''\). So \(\aut{B}\)
can be enriched to generate an infinite group
by~Lemma~\ref{lm:binary}, say
with~\(\rho=(\rho_\ttt:\{i,j\}\to\{i,j\})_{\ttt\in A}\). This group is a
quotient of any group obtained by completing each~\(\rho_\ttt\)
from~\(\{i,j\}\) into~\(\Sigma\), and we can conclude.
\qed
\end{proof}

\subsection*{Acknowledgments} The authors would like to thank Jean
Mairesse who has detected a serious gap in a previous version of this
paper.

\bibliographystyle{splncs_srt}
\bibliography{./extension}

\end{document}